\documentclass[12pt,centertags,reqno]{amsart}

\usepackage[foot]{amsaddr}
\usepackage{mathrsfs, mathtools}
\usepackage{latexsym}
\usepackage[english]{babel}
\usepackage[T1]{fontenc}

\usepackage[utf8]{inputenc}
\usepackage{amssymb}
\usepackage[numbers]{natbib}
\usepackage{url}
\usepackage{hyperref}
\usepackage{verbatim}
\usepackage{graphicx}
\usepackage{leftidx}
\usepackage{mathrsfs, mathtools}
\mathtoolsset{showonlyrefs}

\numberwithin{equation}{section} \makeatletter
\renewcommand{\subsection}{\@startsection
{subsection}{2}{0mm}{\baselineskip}{-0.25cm}
{\normalfont\normalsize\bf}} \makeatother

\newtheorem{theorem}{Theorem}[section]
\newtheorem{lemma}[theorem]{Lemma}
\newtheorem{corollary}[theorem]{Corollary}
\newtheorem{remark}[theorem]{Remark}

\def \F {\mathcal F}

\def \N {\mathcal N}

\def \P {\mathbf P}
\def \Q {\mathbf Q}
\def \R {\mathbb R}
\def \E {\mathcal E}
\def \I {{\mathbf 1}}
\def \bF {\mathbb F}

\newcommand*\xbar[1]{%
  \hbox{%
    \vbox{%
      \hrule height 0.5pt 
      \kern0.5ex
      \hbox{%
        \kern-0.1em
        \ensuremath{#1}%
        \kern-0.1em
      }%
    }%
  }%
}

\newcommand{\ud}{\mathrm d}
\newcommand{\ds}{\displaystyle}
\newcommand{\esp}[2][\mathbb E] {#1\left[#2\right]}
\newcommand{\espp}[2][\mathbb E^{\widehat \P}] {#1\left[#2\right]}
\newcommand{\espbar}[2][\mathbb E^{\bar \P}] {#1\left[#2\right]}

\newcommand{\doleans}[1] {\mathcal E\left(#1\right)}

\begin{document}

\date{}

\author[A.~Cretarola]{Alessandra Cretarola}
\address{Alessandra Cretarola, Department of Mathematics and Computer Science,
 University of Perugia, Via Luigi Vanvitelli, 1, I-06123 Perugia, Italy.}\email{alessandra.cretarola@unipg.it}

\author[G.~Figà Talamanca]{Gianna Figà Talamanca}
\address{Gianna Figà Talamanca, Department of Economics,
University of Perugia, Via Alessandro Pascoli,
I-06123 Perugia, Italy.}\email{gianna.figatalamanca@unipg.it}


%
%
\title[A confidence-based model for asset and derivative prices in the BitCoin market
]{A confidence-based model for asset and derivative prices in the BitCoin market}

\begin{abstract}
We endorse the idea, suggested in recent literature, that BitCoin prices are influenced by sentiment and confidence about the underlying technology; as a consequence, an excitement about the BitCoin system may propagate to  BitCoin prices causing a Bubble effect, the presence of which is documented in several papers about the cryptocurrency. 
In this paper we develop a bivariate model in continuous time to describe the price dynamics of one BitCoin as well as the behavior of a second factor affecting the price itself, which we name confidence indicator.
The two dynamics are possibly correlated and we also take into account a delay between the confidence indicator and its delivered effect on the BitCoin price. 
Statistical properties of the suggested model are investigated and
its arbitrage-free property is shown.
Further, based on risk-neutral evaluation, a quasi-closed formula is derived for European style derivatives on the BitCoin. A short numerical application is finally provided.
\end{abstract}

\maketitle


\section{Introduction}

The BitCoin was first introduced as an electronic payment system between peers. It is based on an open source software which generates a peer to peer network. This network includes a high number of computers connected to each other through the Internet and complex mathematical procedures are implemented both to check the truthfulness of the transaction and to generate new BitCoins. Opposite to traditional transactions, which are based on the trust in financial intermediaries, this system relies on the network, on the fixed rules and on cryptography. The open source software was created in 2009 by a computer scientist known under the pseudonym Satoshi Nakamoto, whose identity is still unknown. BitCoin has several attractive properties for consumers: it does not rely on central banks to regulate the money supply and it enables essentially anonymous transactions. Besides, transactions are irreversible and can also be very small. BitCoins can be purchased on appropriate websites that allow to change usual currencies in BitCoins. Further, payments can be made in BitCoins for several online services and goods and its use is increasing. Special applications have been designed for smartphones and tablets for transactions in BitCoins and some ATM have appeared all over the world (see Coin ATM radar) to change traditional currencies in BitCoins. At very low expenses it is also possible to send cryptocurrency internationally. 
Indeed, BitCoin have experienced a rapid growth both in value and in the number of transactions as shown in Figure \ref{BitCoin}.
\begin{figure}[htbp]
\includegraphics[width=13cm, height=6cm]{
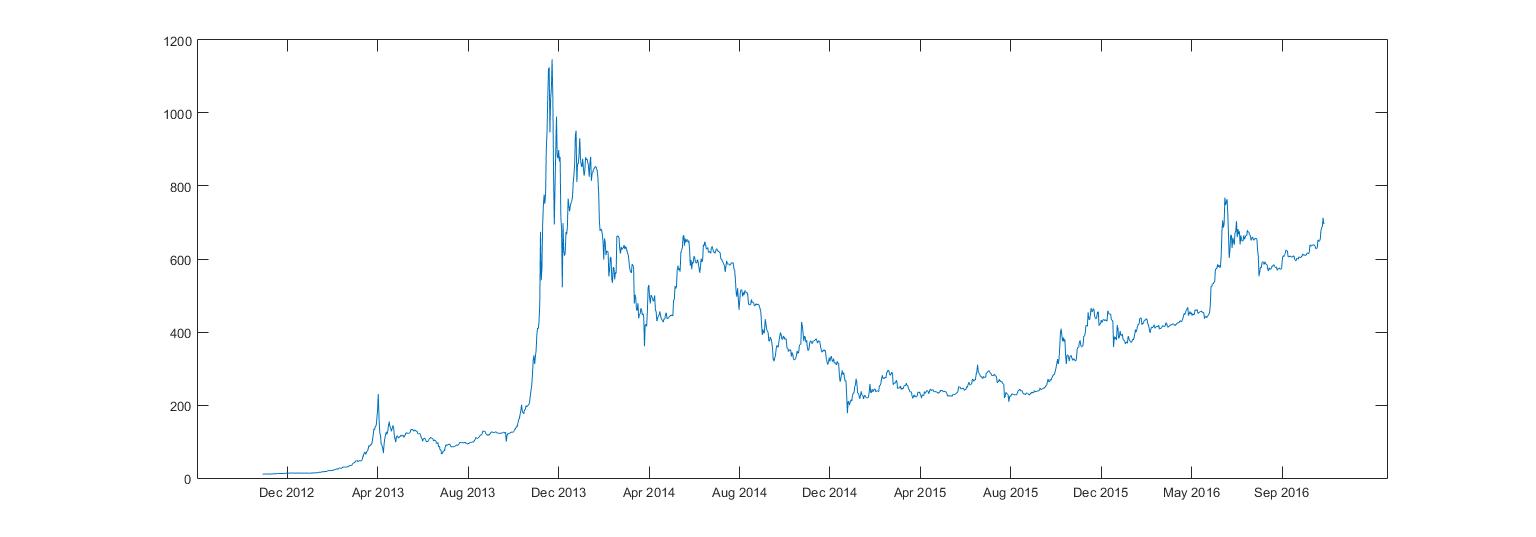}
\caption{Daily average BitCoin price (Dec 2012 - Sep 2016).} \label{BitCoin}
\end{figure}
Although the popularity of BitCoins has increased significantly, cryptocurrencies still face important issues. One of the main issues about BitCoin is whether it should be considered a currency, a commodity or a stock. In \citet{Yermack}, the author performs a detailed qualitative analysis of BitCoin price behavior. He remarks that a currency is usually characterized by three properties: a medium of exchange, a unit of account and a store of value. BitCoin is indeed a medium of exchange, though limited in relative volume of transactions and essentially restricted to online markets; however it lacks the other two properties. BitCoin value is rather volatile and traded for different prices in different exchanges, making it unreliable as a unit of account.  Further, BitCoins can be only deposited in a digital wallet which is costly and possibly subject to hacking attacks, thefts and other issues related to cyber-security.
The conclusion in  \citet{Yermack} is that BitCoin behaves as a high volatility stock and that most transactions on BitCoins are aimed to speculative investments. A second issue about BitCoin prices is the possibility of arbitrage given that it is traded on different web-exchanges for different prices; a pioneering theoretical contribution on this topic is given by \citet{doi:10.1080/14697688.2016.1231928}.

In recent years several papers have appeared in order to model BitCoin price behavior. Interesting results in discrete time are obtained in \citet{HencGou} where the authors  model BitCoin price changes and market bubbles through a non causal econometric model. 

Besides, many authors claim that the high volatility in BitCoin prices and the occurrence of speculative bubbles depend on positive sentiment and confidence about the BitCoin market itself: of course confidence on BitCoin or, more generally, on cryptocurrencies or IT finance is not directly observed  but several variables may be considered as indicators, from the more traditional volume or number of transactions to the number of Google searches or Wikipedia requests about the topic, in the period under investigation. Main references in this area are \citet{MainDrivers, GoogleTrends, SentimentAnalysis}. Alternatively, in \citet{bukovina2016sentiment} confidence is measured by sentiment related to the BiCoin system and made available from the website Sentdex.com. This website collect data on sentiment through an algorithm, based on Natural Language Processing techniques,  which is capable of identifying string of words conveying positive, neutral or negative sentiment on a topic (BitCoin in this case). The authors of the paper develop a model in discrete time and show that positive sentiment, that is excessive confidence on the system, may indeed boost a Bubble on the BitCoin price.   
We borrow the idea from the quoted papers and develop a bivariate model in continuous time to describe both the dynamics of a BitCoin confidence indicator and of the corresponding BitCoin price. We also account for a possible delay between this indicator and its effect on BitCoin prices. 

In this paper we focus on theoretical properties for the suggested model; the choice for the most suitable confidence index among the ones proposed in the literature, of course crucial for applications of our model, is postponed to future research.
In particular, after analyzing some statistical properties of the model, we give conditions under which the model is arbitrage-free and, based on risk-neutral evaluation, we derive a quasi-closed formula for European style derivatives on the BitCoin.
It is worth noticing that a market for this contingent claims has recently raised on appropriate websites such as  https://coinut.com,  trading European Calls and Puts as well as Binary options, i.e. pure bets on the BitCoin price, endorsing the idea in \citet{Yermack} that BitCoins are used for speculative purposes.

The rest of the paper is structured as follows. In Section 2 we describe the model for the BitCoin price dynamics and derive its statistical properties. In Section 3 we prove a quasi-closed formula for European-style derivatives with detailed computations for Plain Vanilla and Binary option prices. Section 4 is devoted to a numerical application and Section 5 gives concluding remarks and hints for future investigations. Most technical proofs are collected in the Appendix.

\section{The BitCoin market model}

We fix a probability space $(\Omega,\F,\P)$ endowed with a filtration 
$\bF = \{\F_t,\ t \ge 0\}$ that satisfies the usual conditions of right-continuity and completeness. 
On the given probability space, we consider a main market in which heterogeneous agents buy or sell BitCoins and denote by 
$S = \{S_t,\ t \geq 0\}$ the price process
of the cryptocurrrency. We assume that the BitCoin price dynamics is described by 
the following equation:
\begin{equation} \label{eq:S}
\ud S_t = \mu_S P_{t-\tau} S_t \ud t+\sigma_S \sqrt{P_{t-\tau}}S_t\ud W_t,\quad  S_0=s_0 \in \R_+,
\end{equation}
where $\mu_{S} \in \R \setminus \{0\}$, $\sigma_{S}\in \R_+$, $\tau \in \R_+$ represent model parameters;  
$W = \{W_t,\ t \ge 0\}$ is a standard Brownian motion on $(\Omega,\F,\P)$, which is $\bF$-adapted, and
$P = \{P_t,\ t \geq 0\}$ is a stochastic factor, representing the confidence or sentiment index in the BitCoin market, satisfying
\begin{equation}\label{eq:BSdyn}
\ud P_t =\mu_P P_t\ud t+\sigma_P P_t\ud Z_t, 
\quad P_t = \phi(t),\ t \in [-L,0].
\end{equation}
Here, $\mu_P \in \R \setminus \{0\}$, $\sigma_P \in \R_+$, $Z = \{Z_t,\ t \geq 0\}$ is a standard Brownian motion on $(\Omega,\F,\P)$ adapted to $\bF$, possibly correlated with $W$, so that $\ud \langle W,Z\rangle_t = \rho \ud t$, for some constant $\rho \in [0,1]$, 
and $\phi:[-L,0] \to [0,+\infty)$ is a continuous (deterministic) initial function. Note that, the nonnegative property of the function $\phi$ corresponds to require that the minimum confidence level is zero.
It is worth noticing that in \eqref{eq:BSdyn} we also consider the effect of the past, since we assume that the confidence index $P$ affects explicitly the BitCoin price $S_t$ up to a certain preceding time $t-\tau$. Assuming that $\tau<L$ and that factor $P$ is observed in the period $[-L,0]$ makes the biviariate model jointly feasible.

It is well-known that the solution of \eqref{eq:BSdyn} is available in closed form and that $P_t$ has a lognormal distribution for each $t > 0$, see \citet{black1973pricing}. 


In order to visualize the dynamics implied by the model in equations \eqref{eq:S} and \eqref{eq:BSdyn}, we plot in Figure \ref{CambioRho} a possible simulated path of daily observations for the confidence process $P$ and the corresponding BitCoin prices $S$ within one year horizon for different levels of correlation; the BitCoin dynamics is represented for the case of independent and perfectly correlated Brownian motions as well as for the case of $\rho=0.5$.
Increasing the correlations gives a boost to the BitCoin market since it strengthens further the dependence between confidence index and price. In Figure \ref{KSdensity} the empirical density of the BitCoin price, estimated with the Kernel smoothing method, is plotted for different values of the correlation between the BitCoin price itself and the index $P$. The suggested dynamics appears to capture the increase in BitCoin volatility given by confidence/sentiment driven trading as observed in \citet{GoogleTrends, bukovina2016sentiment}. 

\begin{figure}[htbp]
\includegraphics[width=13cm, height=6cm]{
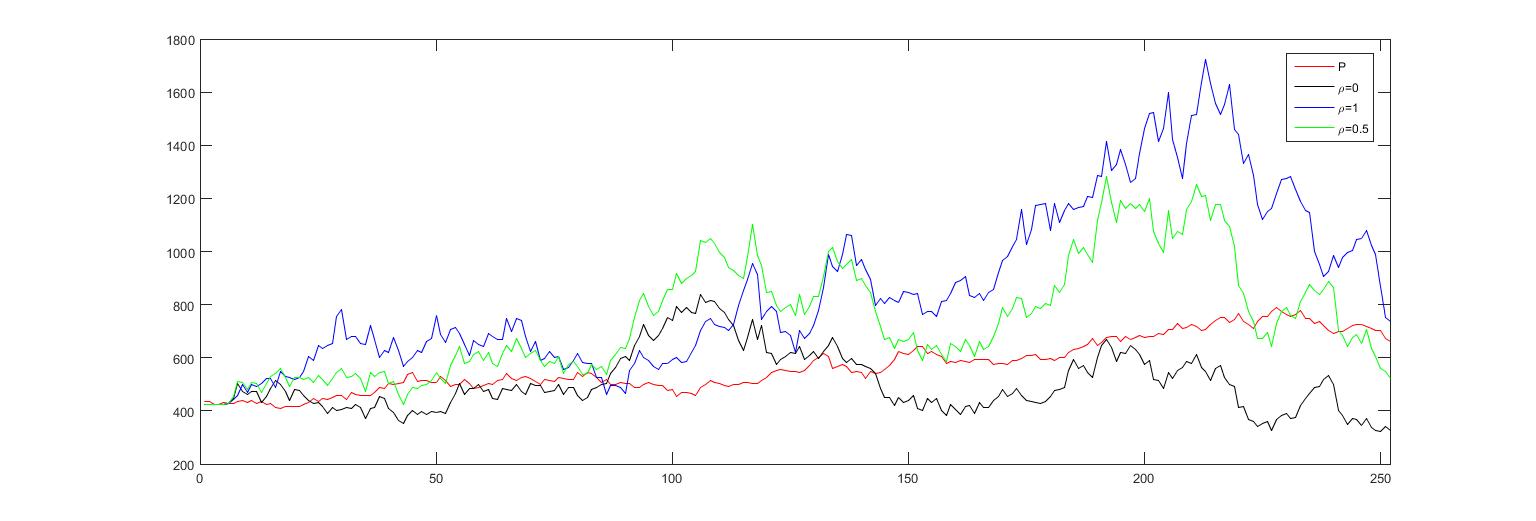}
\caption{An example of BitCoin price dynamics given the evolution of the confidence index (red): $\rho=0$ (black), $\rho=0.5$ (blue), $\rho=0.5$  (green).
Model parameters are set to $\mu_P=0.03,\sigma_P=0.35$, $\mu_S=10^{-5},\sigma_S=0.04$, $\tau=1$ week.} \label{CambioRho}
\end{figure}

\begin{figure}[htbp]
\includegraphics[width=13cm, height=6cm]{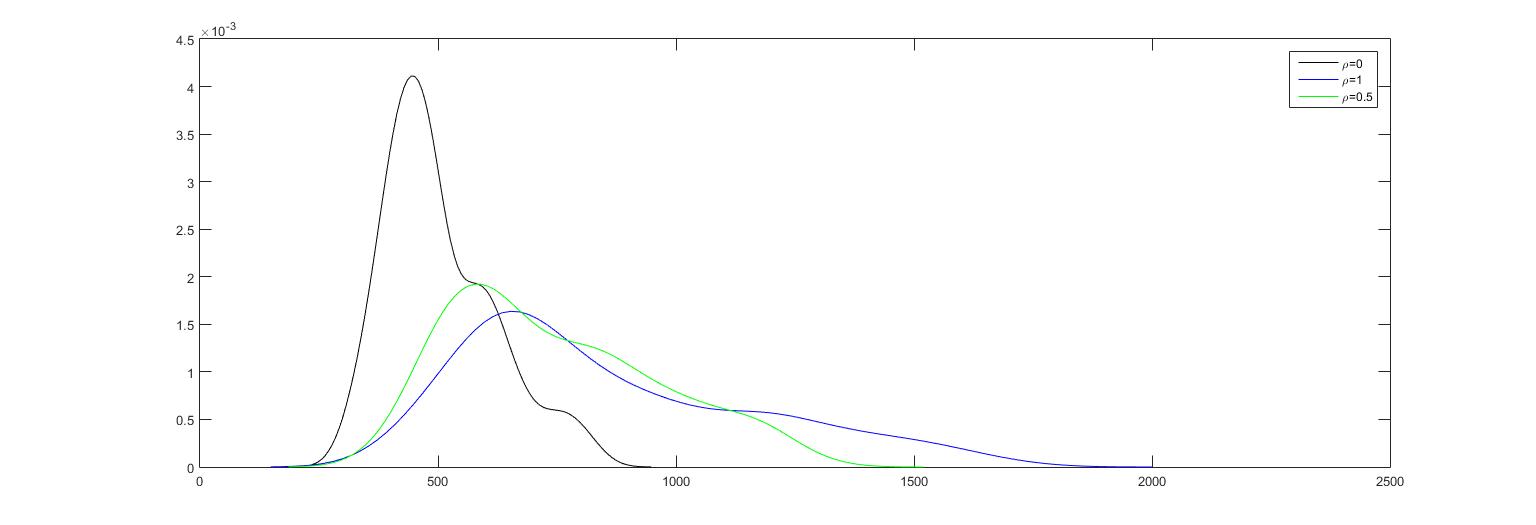} 
\caption{Estimated empirical density for simulated values of BitCoin price according to different correlation parameters with the confidence index process:  $\rho=0$ (black), $\rho=1$ (blue), $\rho=0.5$  (green). Model parameters are set to $\mu_P=0.03,\sigma_P=0.35$, $\mu_S=10^{-5},\sigma_S=0.04$ and $\tau=1$ week.}\label{KSdensity}
\end{figure}
Nevertheless, the shape of the paths for different correlation parameters are similar. This motivates the focus on the non correlated case in the rest of the paper. This means that
$W$ and $Z$ turn out to be independent $(\bF,\P)$-Brownian motions in the underlying market model. Further research will be devoted to the general case in the next future. 

In Figure \ref{CambioTau} the same paths are reported letting $\tau$ vary; as expected market reaction to sentiment is delayed when $\tau$ increases. 

\begin{figure}[htbp]
\includegraphics[width=13cm, height=6cm]{
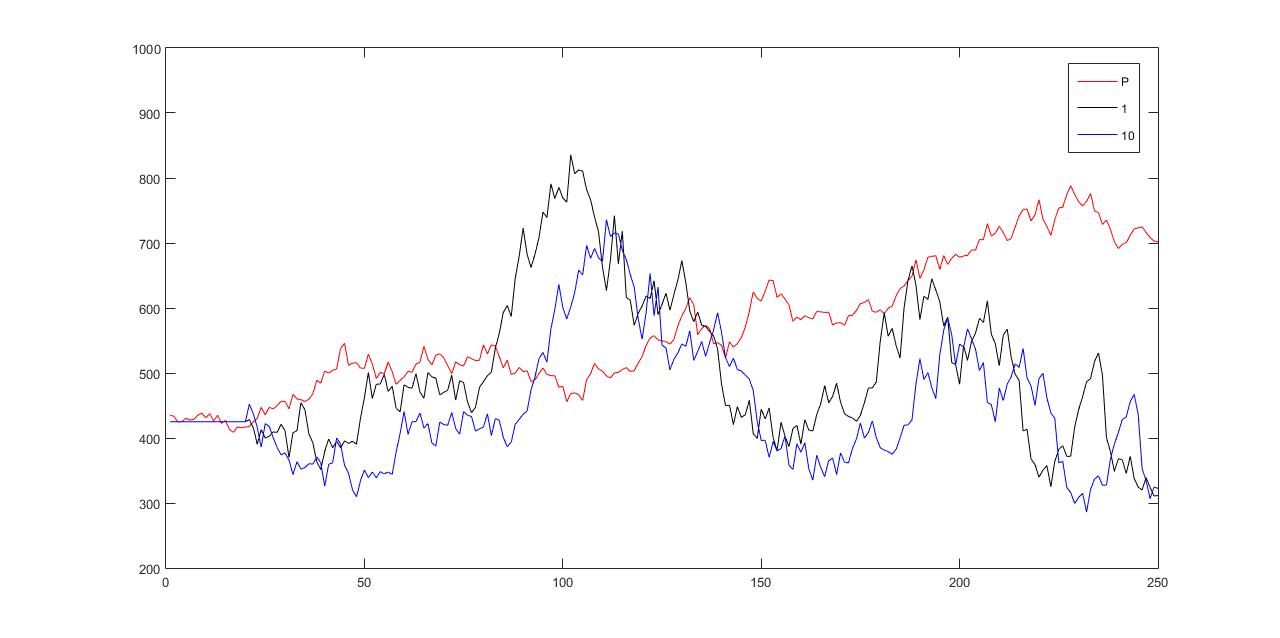} 
\caption{An example of BitCoin price dynamics given the evolution of the confidence index (red) in the independent case $\rho=0$: $\tau=1$ day (black), $\tau=2$ weeks (blue). Model parameters are set to $\mu_P=0.03,\sigma_P=0.35$, $\mu_S=10^{-5},\sigma_S=0.04$.}\label{CambioTau}
\end{figure}
We assume that the reference filtration $\bF=\{\F_t,\ t \geq 0\}$, describing the information on the BitCoin market, is of the form
$$
\F_t=\F_t^W \vee \F_t^Z, \quad t \ge 0,
$$
where $\F_t^W$ and $\F_t^Z$ denote the $\sigma$-algebras generated by $W_t$ and $Z_t$ respectively up to time $t \geq 0$. Note that $\F_t^Z=\F_t^P$, for each $t \geq 0$. Moreover, we suppose that the available information to traders
is described by the filtration $\widetilde \bF=\{\widetilde \F_t,\ t \geq 0\}$, defined by
$$
\widetilde \F_t = \F_t^W  \vee \F_{t-\tau}^P, \quad t \geq 0.
$$
We also remark that all filtrations satisfy the usual conditions of completeness and right-continuity (see e.g. \citet{protter2005stochastic}).
Now, set $\bar P_{0,t}:=\int_{0}^{t}P_{u}\ud u$, for each $t > \tau$. 
Then, we can introduce the {\em integrated information process} 
$X^\tau=\{X_t^\tau,\ t \ge 0\}$ associated to the confidence index $P$, as follows:
\begin{equation} \label{eq:int_info}
X_t^\tau  := \int_0^t P_{u-\tau} \ud u =\int_{-\tau}^{0} \phi(u) \ud u + \int_0^{t-\tau} P_u \ud u= X_\tau^\tau + \int_0^{t-\tau} P_u \ud u, \quad t \ge 0.
\end{equation}
Note that, $X_t^\tau$ turns out to be deterministic for $t \in [0,\tau]$. Indeed,
\begin{equation} \label{eq:int2_info}
\int_0^t P_{u-\tau} \ud u  = \left\{ 
\begin{array}{ll}
0, & \quad  \mbox{if}\ t=0,\\
\int_{-\tau}^{t-\tau} \phi(u) \ud u = X_\tau^\tau + \int_0^{t-\tau} \phi(u) \ud u, & \quad t \in (0,\tau], \\
\int_{-\tau}^{0} \phi(u) \ud u + \int_0^{t-\tau} P_u \ud u = X_\tau^\tau + \bar P_{0,t-\tau}, & \quad  \mbox{if}\ t > \tau.
\end{array}
\right.
\end{equation}
The following lemma establishes basic statistical properties of the integrated information process $X_t^\tau$, for $t>\tau$. 
\begin{lemma}\label{th:means}
In the market model outlined above, for every $t > \tau$, 
we have:
\begin{align*}
\esp{X_t^\tau} & = 
X_\tau^\tau +\frac{\phi(0)}{\mu_{P}}\left(\exp\left(\mu_{P}(t-\tau)\right) -1\right);\\
{\rm Var}[X_t^\tau] & = 
\frac{2\phi^{2}(0)}{\left(\mu_{P}+\sigma_{P}^{2}\right)\left(2\mu_{P}+\sigma_{P}^{2}\right)}\left[\exp\left(\left(2\mu_{P}+\sigma_{P}^{2}\right)(t-\tau)\right) -1\right]\\
& \qquad -\frac{2\phi^{2}(0)}{\mu_{P}\left(\mu_{P}+\sigma_{P}^{2}\right)}\left(\exp\left(\mu_{P}(t-\tau)\right) -1\right)-\left(\frac{\phi(0)}{\mu_{P}}\left(\exp\left(\mu_{P}(t-\tau)\right) -1\right)\right)^2.
\end{align*}
For $t\in [0,\tau]$, we get $\esp{X_t^\tau}=\int_{-\tau}^{\tau-t} \phi(u) \ud u$ and ${\rm Var}[X_t^\tau]=0$.
\end{lemma}
The proof is postponed to Appendix \ref{appendix:technical}.\\
The system given by equations \eqref{eq:S} and \eqref{eq:BSdyn}
is well-defined in $\R_+$ as stated in the following theorem, which also provides its explicit solution.
\begin{theorem} \label{th:sol}
In the market model outlined above, the followings hold:
\begin{itemize}
\item[(i)] the bivariate stochastic delayed differential equation 
\begin{equation}\label{eq:Bivdyn}
\left\{
\begin{array}{ll}
\ud S_t = \mu_S P_{t-\tau} S_t \ud t+\sigma_S \sqrt{P_{t-\tau}}S_t\ud W_t, \quad  S_0=s_0 \in \R_+,  \\
\ud P_t =\mu_P P_t\ud t+\sigma_P P_t\ud Z_t, \quad P_t = \phi(t),\ t \in [-L,0], \\
\end{array}
\right.
\end{equation}
has a continuous, $\bF$-adapted, unique solution $(S,P)=\{(S_t,P_t),\ t \geq 0\}$ given by
\begin{align}
S_t & =s_0\exp\left(\left(\mu_{S}-\frac{\sigma_{S}^{2}}{2}\right)\int_0^t P_{u-\tau} \ud u+\sigma_S \int_0^t \sqrt{P_{u-\tau}}\ud W_u\right),\quad t \ge 0,\label{eq:sol_S}\\
P_t & = \phi(0)\exp\left(\left(\mu_{P}-\frac{\sigma_{P}^{2}}{2}\right)t+\sigma_P Z_t\right), \quad t \ge 0. \label{eq:sol_P}
\end{align}
More precisely, $S$ can be computed step by step as follows: for $k=0,1,2,\ldots$ and $t \in [k\tau,(k+1)\tau]$,
\begin{equation} \label{eq:sol_S1} 
S_t  = S_{k\tau}\exp\left(\left(\mu_{S}-\frac{\sigma_{S}^{2}}{2}\right)\int_{k\tau}^t P_{u-\tau} \ud u+\sigma_S \int_ {k\tau}^t \sqrt{P_{u-\tau}}\ud W_u\right). 
\end{equation}
In particular, 
$P_t \ge 0$ $\P$-a.s. for all $t \geq 0$. If in addition, $\phi(0) > 0$, then $P_t > 0$ $\P$-a.s. for all $t \geq 0$.
\item [(ii)] Further, for every $t \ge 0$, the conditional distribution of $S_{t}$, given the integrated information $X_t^\tau$,
is Log-Normal with mean $\log \left(s_0\right) + \left(\mu_{S}-\frac{\sigma_{S}^{2}}{2}\right)X_t^\tau$ and variance $\sigma_{S}^{2}X_t^\tau$.
\item [(iii)] Finally, for every $t \in [0,\tau]$, the random variable $\log \left( S_t \right)$ has mean $\log \left(s_0\right) + \left(\mu_{S}-\frac{\sigma_{S}^{2}}{2}\right)X_t^\tau$ and variance $\sigma_{S}^{2}X_t^\tau$; for every $t > \tau$,  $\log \left(S_t\right)$ has mean and variance respectively given by
\begin{align*}
\esp{\log \left( S_t \right)} & = \log \left(s_0\right) + \left(\mu_S-\frac{\sigma_S^2}{2}\right)  \esp{X_t^\tau};\\
{\rm Var}\left[{\log \left( S_t \right)}\right] & = \left(\mu_S-\frac{\sigma_S^2}{2}\right)^2{\rm Var}[X_t^\tau]+ \sigma_S^2 \esp{X_t^\tau},
\end{align*}
where $\esp{X_t^\tau}$ and ${\rm Var}[X_t^\tau]$ are both provided by Lemma \ref{th:means}.
\end{itemize}
\end{theorem}

\begin{proof}
{\bf Point (i)}. 
Clearly, $S$ and $P$, given in \eqref{eq:sol_S} and \eqref{eq:sol_P} respectively, are $\bF$-adapted processes with continuous trajectories.
Similarly to \citet[Theorem 2.1]{mao2013delay}, we provide existence and uniqueness of a strong solution to the pair of stochastic differential equations in
system \eqref{eq:Bivdyn} by using forward induction steps of length $\tau$,
without the need of checking any additional assumptions on the coefficients, e.g. the local Lipschitz condition and the linear growth condition.

First, note that the second equation in the system \eqref{eq:Bivdyn} does not depend on $S$, and its solution is well known for all $t \ge 0$. Clearly, equation \eqref{eq:sol_P} says that $P_t \ge 0$ $\P$-a.s. for all $t \ge 0$ and that $\phi(0) > 0$ implies that the solution $P$ remains strictly greater than $0$ over $[0,+\infty)$, i.e. $P_t > 0$, $\P$-a.s. for all $t \ge 0$.\\
Next, by the first equation in 
\eqref{eq:Bivdyn} and applying It\^o’s formula to $\log \left(S_t\right)$, we get
\begin{equation} \label{eq:logS}
\ud\log \left(S_t\right) =\left(\mu_S-\frac{\sigma_S^2}{2}\right)P_{t-\tau} \ud t + \sigma_S\sqrt{P_{t-\tau}} \ud W_t,
\end{equation}
or equivalently, in integral form
\begin{equation} \label{eq:log_int}
\log\left(\frac{S_{t}}{s_{0}}\right) = \left(\mu_S-\frac{\sigma_S^2}{2}\right)\int_0^t P_{u-\tau} \ud u + \sigma_S\int_0^t\sqrt{P_{u-\tau}} \ud W_u,\quad t \ge 0.
\end{equation}
For $t \in [0,\tau]$, \eqref{eq:log_int} can be written as
\begin{equation} \label{eq:log1}
\log\left(\frac{S_{t}}{s_{0}}\right) = \left(\mu_S-\frac{\sigma_S^2}{2}\right)\int_0^t \phi \left(u-\tau \right) \ud u + \sigma_S\int_0^t\sqrt{\phi \left(u-\tau \right)} \ud W_u,
\end{equation}
that is, \eqref{eq:sol_S1} holds for $k = 0$.

Given that $S_t$ is now known for $t \in [0,\tau]$, we may restrict the first equation in \eqref{eq:Bivdyn} on $t \in [\tau, 2\tau]$, so that it corresponds to consider \eqref{eq:logS} for $t \in [\tau, 2\tau]$.
Equivalently, in integral form,
\begin{equation} \label{eq:log}
\log\left(\frac{S_{t}}{S_{\tau}}\right) = \left(\mu_S-\frac{\sigma_S^2}{2}\right)\int_\tau^t P_{u-\tau} \ud u + \sigma_S\int_\tau^t\sqrt{P_{u-\tau}} \ud W_u.
\end{equation}
This shows that \eqref{eq:sol_S1} holds for $k = 1$.
Similar computations for $k=2,3,\ldots$, give the final result.

{\bf Point (ii)}. 
Set  $Y_t:=\int_0^t\sqrt{\phi \left( t-\tau \right)} \ud W_u $, for $t \in [0,\tau]$ and $Y_t:= Y_{k\tau}+ \int_{k\tau}^t\sqrt{P_{u-\tau}} \ud W_u $, for $t \in (k\tau,(k+1)\tau]$, with $k=1,2,\ldots$. Then, by applying the outcomes in Point (i) and the decomposition
$$
\log\left(\frac{S_{t}}{s_{0}}\right) =\log\left(\frac{S_{t}}{S_{k\tau}}\right)+\sum_{j=0}^{k-1}\log\left(\frac{S_{(j+1)\tau}}{S_{j\tau}}\right), 
$$
for $t\in (k\tau,(k+1)\tau]$, with $k=1,2,\ldots$, we can write
\begin{equation} \label{eq:LOG}
\log\left(S_{t}\right)= \log(s_0)+\left(\mu_S-\frac{\sigma_S^2}{2}\right)X_t^\tau +\sigma_S Y_t,\quad t \ge 0.
\end{equation}
To complete the proof, 
it suffices to show that, for each $t\ge 0$ the random variable $Y_t$, conditional on $X_t^\tau$, is Normally distributed with mean $0$ and variance $X_t^\tau$. This is straightforward from \eqref{eq:log1} if $t \in [0,\tau]$. Otherwise, we first observe that
since $Z_{u-\tau}$ is independent of $W_u$ for every $\tau < u \leq t$, the distribution of $Y_t$, conditional on $\{Z_{u-\tau}:\ \tau < u \leq t-\tau\}=\{P_{u-\tau}:\ \tau < u \leq t-\tau\} 
=\F_{t-\tau}^P$, is Normal with mean $0$ and variance $\sigma_S^2X_t^\tau$. \\
Now, for each $t > \tau$, the moment-generating function of $Y_t$, conditioned on the history of the process $P$ up to time $t-\tau$, is given by
\begin{align*}
\esp{e^{aY_t}\Big{|}\F_{t-\tau}^{P}}
& = e^{\int_0^t \frac{a^2}{2} P_{u-\tau} \ud u} =  e^{\frac{a^2}{2}  \int_0^t  P_{u-\tau} \ud u}\\
& = e^{\frac{a^2}{2}\left(\sqrt{X_t^\tau}\right)^2}, \quad  a \in \R,
\end{align*}
that only depends on its integrated information $X_t^\tau$ up to time $t$, that is,  
$$
\esp{e^{aY_t}\Big{|}\F_{t-\tau}^{P}}=\esp{e^{aY_t}\Big{|}X_t^\tau}, \quad t > \tau.
$$

{\bf Point (iii)}. The proof is trivial for $t \in [0,\tau]$. If $t >\tau$, \eqref{eq:LOG} and Lemma \ref{th:means} together with the null-expectation property of the It\^o integral, 
give
\begin{align*}
\esp{\log \left(S_t\right)} & = \log \left(s_0\right) + \left(\mu_S-\frac{\sigma_S^2}{2}\right) \esp{X_t^\tau}\\
& = \log \left(s_0\right) + \left(\mu_S-\frac{\sigma_S^2}{2}\right)\left(X_\tau^\tau  +\frac{\phi(0)}{\mu_{P}}\left(\exp\left(\mu_{P}(t-\tau)\right) -1\right)\right).
\end{align*}
Now, we compute the variance of $\log \left(S_t\right)$. 
Since for each $t > \tau$ the random variable $Y_t $ has mean $0$ conditional on $\F_{t-\tau}^P$, we have
\begin{align*}
{\rm Var}\left[\log \left(S_t\right)\right] & = \esp{\log^2\left(S_t\right)} - \left(\esp{\log\left(S_t\right)}\right)^2\\
& = \left(\mu_S-\frac{\sigma_S^2}{2}\right)^2\esp{(X_t^\tau)^2}  + 2\left(\mu_S-\frac{\sigma_S^2}{2}\right)\esp{X_t^\tau \esp{Y_t\Big{|}\F_{t-\tau}^P}}\\
& \qquad \qquad + \sigma_S^2 \esp{X_t^\tau} - \left(\mu_S-\frac{\sigma_S^2}{2}\right)^2 \esp{X_t^\tau}^2\\
& = \left(\mu_S-\frac{\sigma_S^2}{2}\right)^2{\rm Var}[X_t^\tau]+ \sigma_S^2 \esp{X_t^\tau}.
\end{align*}

Thus, the proof is complete.
\end{proof}

\section{Existence of a risk-neutral probability measure and Derivative pricing}
Let us fix a finite time horizon $T>0$ 
and assume the existence of a riskless asset, say the money market account, whose value process $B=\{B_t,\ t \geq 0\}$ 
is given by
$$
B_t=\exp{\left(\int_0^t r(s)\ud s\right)},\quad t \geq 0,
$$
where $r:[0,+\infty) \to \R$ is a bounded, deterministic function 
representing the instantaneous risk-free interest rate.
To exclude arbitrage opportunities, we 
need to check that the set of all equivalent martingale measures for the BitCoin price process $S$ is non-empty. More precisely, it contains more than a single element, since $P$ does not represent the price of any tradeable asset, and therefore the underlying market model is incomplete.
\begin{lemma}\label{lem:measure}
Let $\phi(t) > 0$, for each $t \in [-L,0]$, in \eqref{eq:BSdyn}. Then, every equivalent 
martingale measure $\Q$ for $S$ defined on $(\Omega,\F_T)$ has the following density 
\begin{equation}\label{def:Q}
\frac{\ud \Q}{\ud \P}\bigg{|}_{\F_T}=:L_T^\Q, \quad \P-\mbox{a.s.},
\end{equation} 
where $L_T^\Q$ is the terminal value of the $(\bF,\P)$-martingale $L^\Q=\{L_t^\Q,\ t \in [0,T]\}$ given by
\begin{equation} \label{eq:L}
L_t^\Q :=\doleans{-\int_0^\cdot \frac{\mu_S P_{s-\tau}-r(s)}{\sigma_S\sqrt{P_{s-\tau}}} \ud W_s - \int_0^t\gamma_s\ud Z_s}_t, \quad t \in [0,T], 
\end{equation}
for a suitable $\bF$-progressively measurable process $\gamma=\{\gamma_t,\ t \in [0,T]\}$.
\end{lemma}
The proof is postponed to Appendix \ref{appendix:technical}. Here $\E (Y)$ denotes the Doleans-Dade exponential of an $(\bF, \P)$-semimartingale $Y$.

In the rest of the paper, suppose that $\phi(t) > 0$, for each $t \in [-L,0]$, in \eqref{eq:BSdyn}.
Then, Lemma \ref{lem:measure} ensures that the space of equivalent martingale measures for $S$ is described by 
\eqref{eq:L}. More precisely, it is
parameterized by the process $\gamma$
which governs the change of drift of the  $(\bF,\P)$-Brownian motion $Z$. Note that
the confidence index dynamics under $\Q$ in the BitCoin market is given by
$$
\ud P_t  = (\mu_P - \sigma_P \gamma_t)P_t\ud t + \sigma_P P_t \ud Z_t^\Q, \quad P_t = \phi(t),\ t \in [-L,0].
$$
The process $\gamma$ 
can be interpreted as the risk perception associated to the future direction or future possible movements of the BitCoin market.
One simple example of a candidate equivalent martingale measure is the so-called {\em minimal martingale measure} (see e.g. \citet{follmer1991hedging}, \citet{follmer2010minimal}), 
denoted by $\widehat \P$, whose density process $L = \{L_t,\ t \in  [0,T]\}$, is given by
\begin{equation} \label{eq:Lp}
L_t :=\exp{\left(-\int_0^t \frac{\mu_S P_{s-\tau}-r(s)}{\sigma_S\sqrt{P_{s-\tau}}} \ud W_s - \frac{1}{2}\int_0^t\left(\frac{\mu_S P_{s-\tau}-r(s)}{\sigma_S\sqrt{P_{s-\tau}}}\right)^2\ud s\right)}. 
\end{equation}

This is the probability measure which corresponds to the choice $\gamma \equiv 0$ in \eqref{eq:L}. 
Intuitively, under the minimal martingale measure, say $\widehat \P$, 
the drift of the Brownian motion driving the BitCoin price process $S$ is modified to make $S$ into an $(\bF,\widehat \P)$-martingale, while the drift of the Brownian motion which is strongly orthogonal to $S$ is not affected by the change measure from $\P$ to $\widehat \P$. More precisely, under the change of measure from $\P$ to $\widehat \P$, we have two independent $(\bF,\widehat \P)$-Brownian motions $\widehat W=\{\widehat W_t,\ t \in [0,T]\}$ and $\widehat Z=\{\widehat Z_t,\ t \in [0,T]\}$ defined respectively by
\begin{align}
\widehat W_t & := W_t + \int_0^t\frac{\mu_S P_{s-\tau}-r(s)}{\sigma_S\sqrt{P_{s-\tau}}} \ud s,\quad t \in [0,T],\\
\widehat Z_t & := Z_t, \quad t \in [0,T]. \label{def:hat_Z}
\end{align}
Denote by  $\widetilde S_t=\{\widetilde S_t,\ t \in [0,T]\}$ the discounted BitCoin price process 
given by $\widetilde S_t:=\frac{S_t}{B_t}$, for each $t \in [0,T]$.
Then, the discounted BitCoin price $\widetilde S_t$, at any time $t \in [0,T]$, is given by
\begin{equation}\label{def:tilde_S}
\widetilde S_t=s_0 \exp{\left(\sigma_S\int_0^t\sqrt{P_{u-\tau}}\ud \widehat W_u - \frac{\sigma_S^2}{2}\int_0^t P_{u-\tau} \ud u\right)},\quad t \in [0,T],
\end{equation}
and the behavior of the confidence index $P$ is still described by equation \eqref{eq:BSdyn}.

Let $H=\varphi(S_T)$ be an $\widetilde \F_T$-measurable random variable representing the payoff a European-type contingent claim with date of maturity $T$, which can be traded on the underlying market. Here $\varphi : \R \to \R$ is a
a Borel-measurable function  such that $H$ is integrable under $\widehat \P$. The function $\varphi$ is usually referred to as the {\em contract function}.
The following result provides a risk-neutral pricing formula under the minimal martingale measure $\widehat \P$ for any $\widehat \P$-integrable European contingent claim. Since the martingale measure is fixed, the risk-neutral price agrees with
the arbitrage free price for
those options which can be replicated by investing on the underlying market. 
Set
$X_{t,T}^\tau:=X_T^\tau-X_t^\tau$, for each $t \in [0,T)$, where the process $X^\tau$ is defined in \eqref{eq:int_info} and denote by $\espbar{\cdot\Big{|}\widetilde \F_t}$ the conditional expectation with respect to $\widetilde \F_t$ under the probability measure $\widehat \P$ and so on.
\begin{theorem}\label{th:cont_claim}
Let  $H=\varphi(S_T)$ be the payoff a European-type contingent claim with date of maturity $T$. Then, the risk-neutral price $\Phi_t(H)$ at time $t$ of $H$ is given by
\begin{equation}\label{eq:gen_option}
\Phi_t(H) = \espp{\psi(t,S_t,X_{t,T}^\tau)\Bigg{|} S_t}, \quad t \in [0,T),
\end{equation}
where 
$\psi: [0,T) \times \R_+ \times \R_+ \longrightarrow \R$  is a Borel-measurable function 
such that
\begin{equation}\label{def:psi}
\psi(t,S_t,X_{t,T}^\tau)=B_t\espp{\frac{1}{B_T}G\left(t,S_t,X_{t,T}^\tau,Y_{t,T}\right)\Bigg{|}\F_t^W \vee \F_{T-\tau}^P},
\end{equation}
for a suitable function $G$ depending on the contract such that $G\left(t,S_t,X_{t,T}^\tau,Y_{t,T}\right)$ is $\widehat \P$-integrable.
\end{theorem}
\begin{proof}
For the sake of simplicity suppose that
$\tau<T$ and set $Y_{t,T}:= \int_t^T\sqrt{P_{u-\tau}}\ud \widehat W_u$, for each $t \in [0,T)$. Then, the risk-neutral price $\Phi_t(H)$ at time $t$ of a European-type contingent claim with payoff $H=\varphi(S_T)$ is given by
\begin{align}
\Phi_t(H) & = B_t \espp{\frac{\varphi(S_T)}{B_T}\bigg{|}\tilde \F_t}\\
& = B_t \espp{\espp{\frac{\varphi\left(S_t\exp{\left(\int_t^T r(u)\ud u-\frac{\sigma_S^2}{2}X_{t,T}^\tau+ \sigma_S Y_{t,T}\right)}\right)}{B_T}\Bigg{|}\F_t^W \vee \F_{T-\tau}^P}\Bigg{|}\tilde \F_t}, \label{eq:phi_S}
\end{align}
where
$\espp{\cdot\Big{|}\widetilde \F_t}$ denotes the conditional expectation with respect to $\widetilde \F_t$ under the minimal martingale measure $\widehat \P$. More generally, \eqref{eq:phi_S} can be written as
\begin{equation}\label{eq:G}
\Phi_t(H) = B_t \espp{\espp{\frac{G(t,S_t,X_{t,T}^\tau, Y_{t,T})}{B_T}\Bigg{|}\F_t^W \vee \F_{T-\tau}^P}\Bigg{|}\tilde \F_t},
\end{equation}
for a suitable function $G$ depending on the contract function $\varphi$.
Since the $(\bF,\P)$-Brownian motion $Z$ driving the confidence index $P$ is not affected by the change of measure from $\P$ to $\widehat \P$ by the definition of minimal martingale measure, 
we have that $Z$ is also an $(\bF,\widehat \P)$-Brownian motion independent of $\widehat W$, see \eqref{def:hat_Z}.
Hence, we can apply
the same arguments
used in point (ii) of the proof of Theorem \ref{th:sol}, to get that, for each $t \in [0,T)$, the random variable $Y_{t,T}$ 
conditioned on $\F_{T-\tau}^P$ is Normally distributed with mean $0$ and variance $X_{t,T}^\tau$. Then, we can write (in law) that $Y_{t,T}= \sqrt{X_{t,T}^\tau} \epsilon$, where $\epsilon$ is a standard Normal random variable and this allows to find a function $\psi$ such that \eqref{def:psi} holds, which means that 
the conditional expectation with respect to $\F_t^W \vee \F_{T-\tau}^P$ in \eqref{eq:G}
only depends on $S_t$ and $X_{t,T}^\tau$, for every $t \in [0,T)$. 
Consequently, the risk-neutral price $\Phi_t(H)$ can be written as
\begin{align*}
\Phi_t(H) & = \espp{\psi(t,S_t,X_{t,T}^\tau)\bigg{|}\tilde \F_t}= \espp{\psi(t,S_t,X_{t,T}^\tau)\Bigg{|} S_t},
\end{align*}
where the last equality holds since $S$ is $\widetilde \bF$-adapted and $X_{t,T}^\tau$ is independent of $\widetilde \F_t$, for each $t \in [0,T)$, 
see e.g. \citet[Lemma A.108]{pascucci2011pde}. More precisely, we have
\begin{equation}
\espp{\psi(t,S_t,X_{t,T}^\tau) \bigg{|}\widetilde \F_t} =\espp{\psi(t,S_t,X_{t,T}^\tau) \bigg{|} S_t}=g(S_t),
\end{equation}
where
\begin{equation}
g(s)=\espp{\psi(t,s,X_{t,T}^\tau) \bigg{|} S_t=s},\quad s \in \R_+.
\end{equation}
\end{proof}

\begin{remark}\label{rem:sigma}
It is worth to remark that $\psi(t,S_t,x)$, with $x \in \R_+$, represents the risk-neutral price at time $t \in [0,T)$ of the contract $H=\varphi(S_T)$ in a Black \& Scholes framework, where the constant volatility parameter $\sigma^{BS}$ is defined by 
$$
\sigma^{BS}:=\sigma_S\sqrt{\frac{x}{T-t}}.
$$
This is proved explicitly in Corollary \ref{th:call2} below for the special case of a \textit{plain vanilla} European Call option.
\end{remark}

\subsection{A Black \& Scholes-type option pricing formula}

Let us consider a European Call option with strike price $K$ and maturity $T$ and define the function $C^{BS}$ as follows
\begin{equation} \label{def:pricing_function}
C^{BS}(t,s,x):=s\mathcal N(d_1(t,s,x)) - K\exp\left(-\int_0^t r(u) \ud u\right)\mathcal N(d_2(t,s,x)),
\end{equation}
where
\begin{equation} \label{def:d1}
d_1(t,s,x)=\frac{\log\left(\frac{s}{K}\right) + \int_0^t r(u)\ud u + \frac{\sigma_S^2}{2}x}{\sigma_S \sqrt{x}}
\end{equation}
and $d_2(t,s,x)=d_1(t,s,x)-\sigma_S \sqrt{x}$, or more explicitly
\begin{equation} \label{def:d2}
d_2(t,s,x)=\frac{\log\left(\frac{s}{K}\right) + \int_0^t r(u)\ud u - \frac{\sigma_S^2}{2}x}{\sigma_S \sqrt{x}}.
\end{equation}
Here, $\mathcal N$ stands for the standard Gaussian cumulative distribution function
$$
\mathcal N(y)=\frac{1}{\sqrt{2\pi}} \int_{-\infty}^{y} e^{-\frac{z^2}{2}}\ud z, \quad \forall\ y \in \R.
$$

\begin{corollary} \label{th:call2}
The risk-neutral price $C_t$ at time $t$ of a European Call option written on the BitCoin with price $S$ expiring in $T$ and with strike price $K$ is given by the formula
\begin{equation}\label{eq:call}
C_t=\espp{C^{BS}(t,S_t,X_{t,T}^\tau)\bigg{|} S_t}, \quad t \in [0,T), 
\end{equation}
where 
the function $C^{BS}:[0,T) \times \R_+ \times \R_+ \longrightarrow \R$ is given by 
\eqref{def:pricing_function}
and the functions $d_1$, $d_2$ are respectively given by \eqref{def:d1}-\eqref{def:d2}. 
\end{corollary}

\begin{proof}
As in the proof of Theorem \ref{th:cont_claim}, let us assume that %
$\tau<T$. Under the 
minimal martingale measure $\widehat \P$,
the risk-neutral price $C_t$ at time $t \in [0,T)$ of a European Call option written on the BitCoin with price $S$ expiring in $T$ and with strike price $K$, is given by
\begin{align*}
C_t & = B_t \espp{\frac{\max \left( S_T-K,0\right)}{B_T}\bigg{|}\widetilde \F_t}\\
& = B_t\espp{\widetilde S_T\I_{\{S_T>K\}}\Big{|}\widetilde \F_t}-K \exp{\left(-\int_t^T r(u) \ud u\right)}\espp{\I_{\{S_T>K\}}\Big{|}\widetilde \F_t}\\
& =: B_t J_1 - K \exp{\left(-\int_t^T r(u) \ud u\right)}J_2,
\end{align*}
Recall that $Y_{t,T}= \int_t^T\sqrt{P_{u-\tau}}\ud \widehat W_u$, for every $t \in [0,T)$. Then, the term 
$J_2$ can be written as
\begin{align}
J_2 & = \espp{\espp{\I_{\{S_T>K\}}|\F_t^W \vee \F_{T-\tau}^P}\Big{|}\widetilde \F_t}\nonumber \\
& = 
\espp{\widehat \P\left(S_t\exp{\left(\int_t^T r(u)\ud u-\frac{\sigma_S^2}{2}X_{t,T}^\tau+\sigma_S Y_{t,T} \right)}>K \bigg{|}\F_t^W \vee \F_{T-\tau}^P \right)\bigg{|}\widetilde \F_t}\\
& = \espp{\widehat \P\left(\sigma_S Y_{t,T} > \log\left(\frac{K}{S_t}\right)- \int_t^T r(u)\ud u + \frac{\sigma_S^2}{2}X_{t,T}^\tau \bigg{|}\F_t^W \vee \F_{T-\tau}^P\right) \bigg{|}\widetilde \F_t}\\
& = 
\espp{\widehat \P\left(-\frac{Y_{t,T}}{
\sqrt{X_{t,T}^\tau}} < \frac{\log\left(\frac{S_t}{K}\right)+ \int_t^T r(u)\ud u - \frac{\sigma_S^2}{2}X_{t,T}^\tau}{
\sigma_S \sqrt{X_{t,T}^\tau}} \bigg{|}\F_t^W \vee \F_{T-\tau}^P\right) \bigg{|}\widetilde \F_t}\\
& = 
\espp{\mathcal N\left( d_2(t,S_t,X_{t,T}^\tau) \right) \bigg{|}\widetilde \F_t}, \nonumber 
\end{align}
as for each $t \in [0,T)$,
the random variable $\ds -\frac{Y_{t,T}}{\sqrt{X_{t,T}^\tau}}$ has a standard Gaussian law $\mathcal N(0,1)$ given $\F_t^W \vee \F_{T-\tau}^P$ under the minimal martingale measure $\widehat \P$.
Concerning $J_1$, 
consider the auxiliary probability measure $\bar \P$ on $(\Omega,\F_T)$ 
defined 
as follows:
\begin{equation} \label{def:barP}
\frac{\ud \bar \P}{\ud \widehat \P} := 
\exp{\left(-\frac{\sigma_S^2}{2} \int_0^T P_{u-\tau}\ud u + \sigma_S \int_0^T \sqrt{P_{u-\tau}}\ud \widehat W_u\right)}, \quad \widehat \P-\mbox{a.s.}.
\end{equation}
By Girsanov's Theorem, we get that the process $\bar W=\{\bar W_t,\ t \in [0,T]\}$, defined by
\begin{equation} \label{def:barW}
\bar W_t := \widehat W_t - \sigma_S\int_0^t\sqrt{P_{u-\tau}}\ud u, \quad t \in [0,T],
\end{equation}
follows a standard $(\bF,\bar \P)$-Brownian motion. In addition, using \eqref{def:tilde_S}, we obtain
\begin{equation} \label{def:tildeST}
\widetilde S_T = \widetilde S_t \exp{\left(\sigma_S\int_t^T\sqrt{P_{u-\tau}}\ud \bar W_u + \frac{\sigma_S^2}{2}\int_t^T P_{u-\tau} \ud u\right)},
\end{equation}
for every $t \in [0,T]$. 
Since $S$ is $\widetilde \bF$-adapted, by \eqref{def:tilde_S} and the Bayes formula on the change of probability measure for conditional expectation, for every $t \in [0,T)$ we get
\begin{align}
J_1 & = \espp{\widetilde S_T \I_{\{S_T>K\}}\Big{|}\widetilde \F_t}\nonumber\\
& = \widetilde S_t \espp{\exp{\left(-\frac{\sigma_S^2}{2} X_{t,T}^\tau +\sigma_S Y_{t,T}\right)}\I_{\{S_T>K\}}\bigg{|}\widetilde \F_t}\nonumber\\
& = \widetilde S_t \frac{\espp{\exp{\left(-\frac{\sigma_S^2}{2} X_{T}^\tau + \sigma_S Y_{0,T}\right)}\I_{\{S_T>K\}}\bigg{|}\widetilde \F_t}}{\exp{\left(-\frac{\sigma_S^2}{2} X_{t}^\tau + \sigma_S Y_{0,t}\right)}}\nonumber\\
& = \widetilde S_t \espbar{\I_{\left\{\widetilde S_T > K B_T^{-1}\right\}}\bigg{|}\widetilde \F_t}\nonumber\\
& = \widetilde S_t \espbar{\I_{\left\{\widetilde S_t \exp{\left( \sigma_S \bar Y_{t,T} + \frac{\sigma_S^2}{2}X_{t,T}^\tau\right)} > K \exp{\left(-\int_0^T r(u) \ud u\right)}\right\}}\bigg{|}\widetilde \F_t}\nonumber\\
& = \widetilde S_t \espbar{ \espbar{\I_{\left\{
 \sigma_S \bar Y_{t,T} > \log\left(\frac{K}{S_t}\right)- \int_t^T r(u)\ud u - \frac{\sigma_S^2}{2}X_{t,T}^\tau
\right\}}\bigg{|}\F_t^W \vee \F_{T-\tau}^P} \bigg{|}\widetilde \F_t}\nonumber\\
& = \widetilde S_t \espbar{\bar \P \left(
-\frac{\bar Y_{t,T}}{\sqrt{X_{t,T}^\tau}} < \frac{\log\left(\frac{S_t}{K}\right) + \int_t^T r(u)\ud u + \frac{\sigma_S^2}{2}X_{t,T}^\tau}{\sigma_S \sqrt{X_{t,T}^\tau}}
\bigg{|}\F_t^W \vee \F_{T-\tau}^P\right) \bigg{|}\widetilde \F_t}\nonumber\\
& = \widetilde S_t\espbar{\mathcal N\left(d_1(t,S_t,X_{t,T}^\tau) \right) \bigg{|}\widetilde \F_t}, \label{term2}
\end{align}
with 
$$
d_1(t,S_t,X_{t,T}^\tau) = d_2(t,S_t,X_{t,T}^\tau) + \sigma_S\sqrt{X_{t,T}^\tau}.
$$
In the above computations, analogously to before, we have set $\bar Y_{t,T}:= \int_t^T\sqrt{P_{u-\tau}}\ud \bar W_u$, for each $t \in [0,T)$. Consequently, we have  that $\bar Y_{t,T}$ 
conditional on $\F_{T-\tau}^P$, is a Normally distributed random variable with mean $0$ and variance $X_{t,T}^\tau$, for each $t \in [0,T)$, since $Z$ is not affected by the change of measure from $\widehat \P$ to $\bar \P$.
Indeed,
by the change of numéraire theorem, we have that the probability measure $\bar \P$ turns out to be the minimal martingale measure corresponding to the choice of the BitCoin price process as benchmark.
Further, by applying again
the Bayes formula on the change of probability measure for conditional expectation, 
we get
\begin{align}
J_1 & = \widetilde S_t\espbar{\mathcal N\left(d_1(t,S_t,X_{t,T}^\tau) \right)\bigg{|}\widetilde \F_t}\nonumber\\
& = \widetilde S_t \frac{\espp{\N\left(d_1(t,S_t,X_{t,T}^\tau) \right)\exp{\left(-\frac{\sigma_S^2}{2} \int_0^T P_{u-\tau}\ud u + \sigma_S \int_0^T \sqrt{P_{u-\tau}}\ud \widehat W_u\right)}\bigg{|}\widetilde \F_t}}{\exp{\left(-\frac{\sigma_S^2}{2} \int_0^t P_{u-\tau}\ud u + \sigma_S \int_0^t \sqrt{P_{u-\tau}}\ud \widehat W_u\right)}}\nonumber\\
& = \widetilde S_t \espp{\N\left(d_1(t,S_t,X_{t,T}^\tau) \right)\exp{\left(-\frac{\sigma_S^2}{2} \int_t^T P_{u-\tau}\ud u + \sigma_S \int_t^T \sqrt{P_{u-\tau}}\ud \widehat W_u\right)}\bigg{|}\widetilde \F_t}\nonumber\\
& = \widetilde S_t \espp{\espp{\N\left(d_1(t,S_t,X_{t,T}^\tau) \right)\exp{\left(-\frac{\sigma_S^2}{2} X_{t,T}^\tau + \sigma_S Y_{t,T}\right)}
\bigg{|}\F_t^W \vee \F_{T-\tau}^P }\bigg{|}\widetilde \F_t}\nonumber\\
& = \widetilde S_t \espp{\N\left(d_1(t,S_t,X_{t,T}^\tau) \right)\exp{\left(-\frac{\sigma_S^2}{2} X_{t,T}^\tau \right)}\espp{\exp{\left(\sigma_S Y_{t,T}\right)}
\bigg{|}\F_t^W \vee \F_{T-\tau}^P}\bigg{|}\widetilde \F_t}\nonumber\\
& = \widetilde S_t\espp{\mathcal N\left(d_1(t,S_t,X_{t,T}^\tau) \right)\bigg{|}\widetilde \F_t}, \label{term1}
\end{align}
since the conditional Gaussian distribution of $Y_{t,T}$ gives
$$
\espp{\exp{\left(\sigma_S Y_{t,T}\right)}\bigg{|}\F_t^W \vee \F_{T-\tau}^P}=
\exp{\left(\frac{\sigma_S^2}{2} X_{t,T}^\tau\right)}.
$$
Finally, gathering the two terms \eqref{term1} and \eqref{term2}, for every $t \in [0,T)$
we obtain
\begin{align}
C_t &= B_t\widetilde S_t\espp{\mathcal N\left(d_1(t,S_t,X_{t,T}^\tau) \right)\bigg{|}\widetilde \F_t}-K \exp{\left(-\int_t^T r(u) \ud u\right)}\espp{\mathcal N\left( d_2(t,S_t,X_{t,T}^\tau) \right) \bigg{|}\widetilde \F_t} \nonumber \\
&= S_t\espp{\mathcal N\left(d_1(t,S_t,X_{t,T}^\tau) \right)\bigg{|}\widetilde \F_t}-K \exp{\left(-\int_t^T r(u) \ud u\right)}\espp{\mathcal N\left( d_2(t,S_t,X_{t,T}^\tau) \right) \bigg{|}\widetilde \F_t} \nonumber \\
&=\espp{C^{BS}(t,S_t,X_{t,T}^\tau)\bigg{|}\widetilde \F_t}\nonumber \\
&= \espp{C^{BS}(t,S_t,X_{t,T}^\tau)\bigg{|} S_t}, 
\end{align}
where the last equality follows again from \citet[Lemma A.108]{pascucci2011pde}, since
for each $t \in [0,T)$, $X_{t,T}^\tau$ is independent of $\widetilde \F_t$  and $S_t$ is  $\widetilde \F_t$-measurable.

\end{proof}

It is worth noticing that the option pricing formula \eqref{eq:call} only depends on the distribution of $X_{t,T}^\tau$  which is the same both under measure $\widehat \P$ and $\bar \P$. 
As observed in Remark \ref{rem:sigma}, formula \eqref{eq:call} evaluated in $S_t$ corresponds to the Black \& Scholes price at time $t \in[0,T)$ of a European Call option written on $S$, with strike price $K$ and maturity $T$, in a market where the volatility parameter is given by $\sigma_S \sqrt{\frac{x}{T-t}}$.
Then, for every $t \in [0,T)$ it  may be written as:
\begin{equation} \label{eq:integr}
C_t=\int_0^{+\infty} C^{BS}(t,S_t,x) f_{X_{t,T}^\tau}(x) \ud x,  
\end{equation}
where 
$f_{X_{t,T}^\tau}(x)$ denotes the density function of $X_{t,T}^\tau$, for each $t \in [0,T)$ (if it exists).

Similar formulas can be computed for other European style derivatives as for binary options which, indeed, are quoted in BitCoin markets. For the case of a Cash or Nothing Call, which is essentially a bet of $A$ on the exercise event, the risk-neutral pricing formula is given by
\begin{align}
C_t & =A\exp{\left(-\int_t^T r(u) \ud s\right)}\espp{\mathcal N\left( d_2(t,S_t,X_{t,T}^\tau) \right)\bigg{|} S_t}\\
&= A\exp{\left(-\int_t^T r(u) \ud s\right)} \int_0^{+\infty} \mathcal N\left( d_2(t,S_t,x) \right) f_{X_{t,T}^\tau}(x) \ud x,\quad t \in [0,T).  \label{binarysimple}
\end{align}

The price at time $t$ for a \textit{plain vanilla} European option may also be written as a Black \& Scholes style price:
\begin{equation} 
C_t=S_t Q_1-K \exp{\left(-\int_t^T r(u) \ud s\right)} Q_2 
\end{equation}
where 
\begin{align} 
Q_1=\espp{\mathcal N\left(d_1(t,S_t,X_{t,T}^\tau) \right)\bigg{|}S_t}
&= \int_0^{+\infty} \mathcal N\left( d_1(t,S_t,x) \right) f_{X_{t,T}^\tau}(x) \ud x,\\
\end{align}

and \begin{align} 
Q_2=\espp{\mathcal N\left(d_2(t,S_t,X_{t,T}^\tau) \right)\bigg{|}S_t}
&= \int_0^{+\infty} \mathcal N\left( d_2(t,S_t,x) \right) f_{X_{t,T}^\tau}(x) \ud x.\\
\end{align}
To compute numerically derivative prices by the above formulas, we should compute the  distribution of $X_{t,T}^\tau$, which is not an easy task.

By applying the scaling property of Brownian motion (see e.g. \citet{carr2004bessel}), for every $t \in [0,T]$, we get
\begin{align*}
\int_0^t P_u \ud u & =P_0\int_0^t \exp{\left ( \sigma Z_u+ (\mu-0.5 \sigma^2) u \right) \ud u }\\
& = \frac{4 P_0}{\sigma^2}\int_0^{\sigma^2 t/4 }\exp{\left(2 (Z_v+ (0.5 \mu-0.25 \sigma^2) 4/ \sigma^2 v \right)\ud v)}=\frac{4 P_0}{\sigma^2} A_h^{(m)},
\end{align*}
where $A_h^{(m)}=\int_0^h \exp{\left( 2(Z_v+ m v )\right)}\ud v$ is the so-called Yor process with $h=\sigma^2 t/4 $ and $m= 2\mu/\sigma^2-1$.
The distribution of $X_{0,T}^\tau$ can be thus obtained through the distribution of the process $A_h^{(m)}$, for $h\geq 0$; a rich literature is devoted to this aim, starting from \citet{yor1992some}, e.g. \citet{dufresne2001integral,matsumoto2005exponential,matsumoto2005exponential2}. 
The main financial application of the above outcomes is to Asian options pricing: see, among others, the seminal paper \citet{geman1993bessel} and more recently \citet{carr2004bessel}.
A rigorous application of the outcomes of the quoted papers should provide a pricing formula by computing (at least numerically) the integral in \eqref{eq:integr} or in \eqref{binarysimple} 
but this is beyond the scope of this paper. Several approximations have been given in the literature to the distribution of the integral of a Geometric Brownian motion, among others \citet{levy1992pricing} and \citet{MilPosner}.

\section{Numerical application}

In \citet{levy1992pricing} the author claims that the distribution of the mean integrated Brownian motion $\frac{1}{s} \bar{P}_{0,s}$ can be approximated with a Log-normal distribution with mean $\alpha(s)$ and variance $\nu^2(s)$, at least for suitable values of the model parameters. In particular, by applying a moment matching technique the Log-normal parameters are given by:
$$
\alpha(u)=\log\frac{1}{u} \frac{\esp{\bar{P}_{0,u}}^2}{\sqrt{\left(\esp{\bar{P}_{0,u}^2}\right)}},
$$
$$
\nu^2(u)=\log \frac{\esp{ \bar{P}_{0,u}^2}}{\esp{\bar{P}_{0,u}}^2},
$$
which can be computed in terms of $\mu_P$ and $\sigma_P$, for all $u \in (0,T]$, from the outcomes in Lemma \ref{th:means}. 
Under the assumptions in the quoted paper and assuming $T>\tau$, the approximate distribution of $X_{0,T}^\tau$ can be derived from the approximate distribution of $\frac{1}{T-\tau} \bar{P}_{0,T-\tau}$ as 
$$
f_{X_{0,T}^\tau}(x)=\frac{1}{T-\tau} \mathcal{LN} pdf_{\alpha(T-\tau),\nu^2(T-\tau)} \left(\frac{x}{T-\tau}\right).
$$

where $ \mathcal{LN} pdf_{m,v} $ denotes the probability distribution function of a Log-normal distribution with parameters $m$ and $v$, defined as

$$
\mathcal{LN} pdf_{m,v}(y)=\frac{1}{y\sqrt{2\pi v}} e^{-\frac{(log (y)-m)^2}{2v}}, \quad \forall\ y \in \R^+.
$$

To have a visual evaluation of the appropriateness of the suggested approximation, we plot in Figure \ref{LNhistfit} the empirical distribution of the logarithm of $\frac{1}{T-\tau} \bar{P}_{0,T-\tau}$ , obtained by simulating $10000$ paths for the Geometric Brownian motion with $dt=1/10000$, $P0=100$ $\mu_P=0.03$, $\sigma_P=0.35$ and  $T-\tau=3$ months. In the same picture the Fitted normal distribution is superimposed.
It appears that the approximation is reliable at last for the parameter case under consideration.

Of course other possible approximations are available as, for example, the inverse gamma approximation in \citet{MilPosner}. Since we are interested in giving an example of application of the pricing formulas, further investigation on the optimal approximation method is beyond the scope of our paper. We just add that the inverse gamma approach holds in the limit when $T$ tends to infinity, which is rarely the case in option pricing, especially in the BitCoin market.

\begin{figure}[htbp]
\includegraphics[width=13cm, height=6cm]{
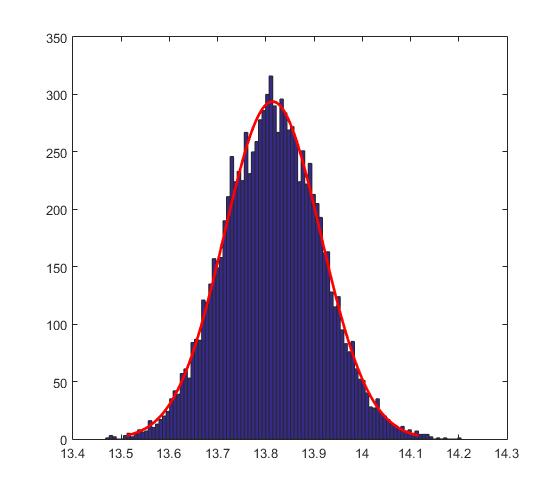} 
\caption{Empirical distribution of the logarithm of $\frac{1}{T-\tau} \bar{P}_{0,T-\tau}$  obtained by simulating 10000 paths for the Geometric Brownian Motion and the corresponding normal fit: $P_0=100$, $\mu_P=0.03$, $\sigma_P=0.35$ and  $T-\tau=3$ months. }\label{LNhistfit}
\end{figure}

In this case the Call option pricing formula becomes
\begin{equation}
C_0 = \frac{1}{T-\tau} \int_0^{+\infty}C^{BS}(0,S_0,x) \mathcal{LN} pdf_{\alpha(T-\tau),\nu^2(T-\tau)} \left(\frac{x}{T-\tau}\right) \ud x,  \label{callexample}
\end{equation}
which can be computed numerically, once parameters $\alpha(T-\tau),\nu(T-\tau)$ are obtained through the above equalities.

\begin{table}[tp]
\caption{Call option prices against different strikes $K$ and for different values of the confidence $P_0$ on BitCoins. Market parameters are $S_0=450,r=0.01,\mu_P=0.03,\sigma_P=0.35, \sigma_S=0.04$, $T=3$ months, $\tau$ = $1$ week (5 days).}
 \label{tab:prices}
 \centering
\begin{tabular}{||c|c|c|c|c|c||}
\hline 
K & 400 & 425 & 450 & 475 & 500 \\ 
\hline 
$P_0=10$ & 51.24 & 28.35 & 11.46 & 3.09 & 0.54 \\ 
\hline 
$P_0=100$ & 64.12 & 48.05 & 34.94 & 24.69 & 16.97 \\ 
\hline 
$P_0=1000$ & 128.68 & 117.75 & 107.77 & 98.66 & 90.35 \\ 
\hline
\end{tabular}
\end{table}

In Table \ref{tab:prices},  Call option prices are reported  in the case where $S_0=450$, $r=0.01$, $\mu_P=0.03$, $\sigma_P=0.35$, $\sigma_S=0.04$, $T=3$ months, $\tau$= 1 week (5 working days). Rows correspond to a different values of the confidence $P_0$ on BitCoins while columns to different values for the strike price. As expected, Call option prices are increasing with respect to confidence in the market and decreasing with respect to strike prices.
\begin{table} 
 \caption{Call option prices against different Strikes K and for different values of $T$ and $tau$. Market parameters are $S_0=450,r=0.01,\mu_P=0.03,\sigma_P=0.35, \sigma_S=0.04$ and $P_0$=100.}
\label{tab:prices2}
\centering
\begin{tabular}{||c|c|c|c|c|c||}

\hline 
K & 400 & 425 & 450 & 475 & 500 \\ 
\hline
$T$=1 month, $\tau$=1 week & 52.85 & 33.09 & 18.27 & 8.81 & 3.71 \\ 
\hline 
$T$=1 month, $\tau$=2 weeks & 51.58 & 30.62 & 15.18 & 6.13 & 2.00 \\ 
\hline 
$T$=3 months, $\tau$=1 week  & 64.12 & 48.05 & 34.94 & 24.69 & 16.97 \\
\hline
$T$=3 months, $\tau$=2 weeks & 62.95 & 46.65 & 33.42 & 23.18 & 15.60 \\ 
\hline
 \end{tabular}
 \end{table}

In Table \ref{tab:prices2}, Call option prices are summed up, for initial confidence value $P_0=100$, by letting the expiration date $T$ and the information lag $\tau$ vary. Again as expected,  for Plain Vanilla Calls the price increases with time to maturity.
Increasing the delay reduces option prices; of course the spread is inversely related to the time to maturity of the option. 
\begin{table}
\caption{Digital Cash or Nothing prices against different Strikes $K$ and for different values of the confidence $P_0$ on BitCoins. Market parameters are $S_0=450$, $r=0.01$, $\mu_P=0.03$, $\sigma_P=0.35$, $\sigma_S=0.04$, $T=3$ months, $\tau= 5$ days. The prize of the option is set to $A=100$.}
\label{tab:prices3}
\centering
\begin{tabular}{||c|c|c|c|c|c||}
\hline 
K & 400 & 425 & 450 & 475 & 500 \\ 
\hline 
$P_0=10$ & 97.17 & 82.77 & 50.31 & 18.87 & 4.24 \\ 
\hline 
$P_0=100$ & 70.07 & 58.38 &46.58 & 35.66 & 26.27 \\ 
\hline 
$P_0=1000$ & 45.70& 41.77 & 38.14 & 34.79 & 31.72 \\ 
\hline
\end{tabular}
 \end{table}

\begin{table} 
 \caption{Digital Cash or Nothing prices against different Strikes $K$ and for different values of $T$ and $\tau$. Market parameters are $S_0=450$, $r=0.01$, $\mu_P=0.03$, $\sigma_P=0.35$, $\sigma_S=0.04$ and $P_0=100$. The prize of the option is set to $A=100$.}
\label{tab:prices4}
\centering 
\begin{tabular}{||c|c|c|c|c|c||}
\hline 
K & 400 & 425 & 450 & 475 & 500 \\ 
\hline 
$T$=1 month, $\tau$=1 week & 86.93 & 69.97 & 48.27 & 28.11 & 13.83 \\ 
\hline 
$T$=1 month, $\tau$=2 weeks & 91.50 & 74.23 & 48.69 & 24.84 & 9.80 \\ 
\hline
$T$=3 months, $\tau$=1 week & 70.07 & 58.38 & 46.58 & 35.66 & 26.27 \\ 
\hline
$T$=3 months, $\tau$=2 weeks & 71.21 & 59.10 & 46.77 & 35.36 & 25.62 \\ 
\hline 

\end{tabular}
\end{table}

In Tables \ref{tab:prices3} and \ref{tab:prices4}, similar results are reported for Binary Options with outcome $A=100$, that are evaluated computing numerically the following integral:

\begin{equation}
C_0=\frac{A}{T-\tau} \int_0^{+\infty} \mathcal N\left( d_2(T-t,S_0,x) \right) \mathcal{LN} pdf_{\alpha(T-\tau),\nu^2(T-\tau)} \left(\frac{x}{T-\tau}\right)\ud x.  
\label{digitalexample}
\end{equation}

Table \ref{tab:prices3}  sums up Binary Cash-or-Nothing prices for $S_0=450$, $r=0.01$, $\mu_P=0.03$, $\sigma_P=0.35$, $\sigma_S=0.04$, $T=3$ months, $\tau= 1$ week (5 working days) against several strikes (in colums). Rows correspond to different values of the confidence $P_0$ on BitCoins. As expected, prices are decreasing with respect to strike prices. Here, \textit{in the money} (ITM) options values are decreasing with respect to $P_0$ while \textit{out of the money} (OTM) ones are increasing. The difference in ITM and OTM prices is large for low values of $P_0$, while it is very small for a high level of confidence in BitCoins. This may be justified by the fact that, when confidence in the BitCoin is strong, all bets are worth, even the OTM ones, since the underlying value is expected to blow up. Binary Call prices decrease with respect to time to maturity for ITM options and increase for OTM options which become more likely to be exercised. The influence of the delay value is tiny, as for vanilla options, being larger for short time to maturities.

\section{Concluding remarks}
In this paper we borrow the idea, suggested in recent literature, that BitCoin prices are boosted by confidence i.e. positive sentiment on the BitCoin system and underlying technology. In particular, we believe that over-confidence may explain the Bubbles documented in several papers about the analysed cryptocurrency. Main references in this area are \citet{MainDrivers, GoogleTrends, SentimentAnalysis,bukovina2016sentiment}. In order to account for such behavior we develop a model in continuous time which describes the dynamics of two factors, one representing the confidence index on the BitCoin system and the other representing the BitCoin price itself, which is directly affected by the first factor.  The two dynamics are possibly correlated and we also take into account a delay between the confidence index and its delivered effect on the BitCoin price. 
A different approach is considered in \citet{HencGou}  where the authors  model BitCoin price and the presence of  market Bubbles through a non causal discrete time model.
We investigate statistical properties of the proposed model and we 
show its arbitrage-free property. By applying the classical risk-neutral evaluation we are able to derive a quasi-closed formula for European style derivatives on the BitCoin with special attention of Plain Vanilla and Binary options for which a market already exists (e.g. https://coinut.com). 
Of course confidence or sentiment about BitCoin or, more generally, on cryptocurrencies or IT finance is not directly observed  but several variables may be considered as indicators, for instance the volume or the number of transactions. Alternatively, more unconventional sentiment indicators may be used as suggested in \citet{GoogleTrends, SentimentAnalysis,bukovina2016sentiment} as the number of Google searches, the number of Wikipedia requests about the topic or indicators  based on Natural Language Processing techniques,  by identifying string of words conveying positive, negative or neutral sentiment about the BitCoin system.
Several open problems are left for future research; first, the choice for the most suitable confidence index among the ones proposed is crucial for applications of our model.
Then, our attention will be also devoted to the full specification model with non-zero correlation between the two factors. In fact, as Figure \ref{CambioRho} suggests, we believe that the model we introduced is capable of describing Bubbles in the BitCoin market by simply modulating the correlation parameter value. 
Last but not least, taking advantage of the statistical properties described in Section 2, we will address the fit of the suggested model to observed data.

\medskip

 \begin{center}
{\bf Acknowledgements}
\end{center}
The authors are grateful to Banca d'Italia and Fondazione Cassa di Risparmio di Perugia for the financial support.

\bibliographystyle{plainnat}
\bibliography{biblio_BS1}

\appendix
\section{Technical proofs}\label{appendix:technical}

\begin{proof}[Proof of Lemma \ref{th:means}]
For reasons of clarity,  we provide a self-contained proof. 
For $t \in [0,\tau]$, the proof is trivial. 
Firstly, given $t > \tau$, we compute the expectation of $X_t^\tau$. 
By relationship \eqref{eq:int_info} and since $P_t >  0$ for each $t > 0$, by applying Fubini’s theorem we get
\begin{align*}
\esp{X_t^\tau} & = X_\tau^\tau + \esp{\bar P_{0,t-\tau}}= X_\tau^\tau + \esp{\int_0^{t-\tau} P_u \ud u}
= X_\tau^\tau + \int_0^{t-\tau} \esp{P_u} \ud u, \quad t > \tau,
\end{align*}
where 
for each $u \ge 0$, we have
\begin{align*}
\esp{P_{u}} & 
 = \phi(0)\exp\left(\left(\mu_{P}-\frac{\sigma_{P}^{2}}{2}\right)u\right)\esp{\exp\left(\sigma_{P}Z_{u}\right)} 
= \phi(0)\exp\left(\mu_{P}u\right),
\end{align*}
since $P$ is a geometric Brownian motion starting from $\phi(0)$.
Hence
$$
\esp{X_t^\tau}=X_\tau^\tau +\phi(0)\int_{0}^{t-\tau}\exp\left(\mu_{P}u\right) \ud u=X_\tau^\tau+\frac{\phi(0)}{\mu_{P}}\left(\exp\left(\mu_{P}(t-\tau)\right) -1\right), \quad t > \tau.
$$
Let us compute the variance of $X_t^\tau$, for each $t > \tau$. 
Given $t > \tau$, we have
\begin{align}
{\rm Var}[X_t^\tau] & 
= {\rm Var}\left[\int_0^{t-\tau} P_u \ud u\right]
= \esp{\left(\int_0^{t-\tau} P_u \ud u\right)^2} - \esp{\int_0^{t-\tau} P_u \ud u}^2,
\end{align}
with
\begin{align}
\esp{\left(\int_0^{t-\tau} P_u \ud u\right)^{2}} & = 
2\esp{\int_{0}^{t-\tau}P_{v}dv\int_{0}^{v-\tau}P_{u}\ud u} 
= 
2\esp{\int_{0}^{t-\tau}\int_{0}^{v-\tau}P_{u}P_{v}\ud v\ud u} \nonumber\\
& = 
2\int_{0}^{t-\tau}\int_{0}^{v-\tau}\esp{P_{u}P_{v}}\ud v\ud u, \label{eq:MP2}
\end{align}
where the last equality holds thanks to Fubini’s theorem. Moreover, by the increments independence property of the Brownian motion, for $0 < u < v \leq t$, we get
\begin{align*}
\esp{P_{u}P_{v}} & = \esp{P_{u}^{2}\exp\left(\left(\mu_{P}-\frac{\sigma_{P}^{2}}{2}\right)(v-u)+\sigma_{P}\left(Z_{v}-Z_{u}\right)\right)}\\
& =\exp\left(\left(\mu_{P}-\frac{\sigma_{P}^{2}}{2}\right)(v-u)\right) \esp{P_{u}^{2}\esp{\exp\left(\sigma_{P}\left(Z_{v}-Z_{u}\right)\right)|\F_u^P}}\\
& =\exp\left(\left(\mu_{P}-\frac{\sigma_{P}^{2}}{2}\right)(v-u)\right) \esp{P_{u}^{2}}\esp{\exp\left(\sigma_{P}\left(Z_{v}-Z_{u}\right)\right)}\\
& =\exp\left(\left(\mu_{P}-\frac{\sigma_{P}^{2}}{2}\right)(v-u)\right) \esp{P_{u}^{2}}\esp{\exp\left(\frac{\sigma_{P}^{2}\left(v-u)\right)}{2}\right)}=\exp\left(\mu_{P}(v-u)\right)\esp{P_{u}^{2}}.
\end{align*}
Further,
\begin{align*}
\esp{P_{u}^{2}} & 
=\phi^2(0)\exp\left(2\left(\mu_{P}-\frac{\sigma_{P}^{2}}{2}\right)u\right)\esp{\exp\left(2\sigma_{P}Z_{u}\right)}
=\phi^2(0)\exp\left(\left(2\mu_{P}+\sigma_{P}^{2}\right)u\right).
\end{align*}
Hence
\begin{equation} \label{eq:pp}
\esp{P_{u}P_{v}}=\phi^2(0)\exp\left(\mu_{P}(v-u)\right)\exp\left(\left(2\mu_{P}+\sigma_{P}^{2}\right)u\right),
\end{equation}
and by plugging \eqref{eq:pp} into \eqref{eq:MP2}, for every $t > \tau$ we have
\begin{align*}
\esp{\left(\int_0^{t-\tau} P_u \ud u\right)^{2}}
& = 
2\phi^2(0)\int_{0}^{t-\tau}\exp\left(\mu_{P}v\right) \int_{0}^{v-\tau}\exp\left(\left(\mu_{P}+\sigma_{P}^{2}\right)u\right) \ud u\ud v\\
&  = \frac{2\phi^{2}(0)}{\left(\mu_{P}+\sigma_{P}^{2}\right)\left(2\mu_{P}+\sigma_{P}^{2}\right)}\left[\exp\left(\left(2\mu_{P}+\sigma_{P}^{2}\right)(t-\tau)\right) -1\right]\\
& \qquad \qquad -\frac{2\phi^{2}(0)}{\mu_{P}\left(\mu_{P}+\sigma_{P}^{2}\right)}\left(\exp\left(\mu_{P}(t-\tau)\right) -1\right).
\end{align*}
Finally, gathering the results we get
\begin{align*}
{\rm Var}[X_t^\tau] & = 
\frac{2\phi^{2}(0)}{\left(\mu_{P}+\sigma_{P}^{2}\right)\left(2\mu_{P}+\sigma_{P}^{2}\right)}\left[\exp\left(\left(2\mu_{P}+\sigma_{P}^{2}\right)(t-\tau)\right) -1\right] \\
& \qquad - \frac{2\phi^{2}(0)}{\mu_{P}\left(\mu_{P}+\sigma_{P}^{2}\right)}\left(\exp\left(\mu_{P}(t-\tau)\right) -1\right)-\left(\frac{\phi(0)}{\mu_{P}}\left(\exp\left(\mu_{P}(t-\tau)\right) -1\right)\right)^2,
\end{align*}
for every $t > \tau$.
\end{proof}

\begin{proof}[Proof of Lemma \ref{lem:measure}]

Firstly, we prove that formula \eqref{eq:L} defines a probability measure $\Q$ equivalent to $\P$ on $(\Omega,\F_T)$. This means we need to show that $L^\Q$ is an $(\bF,\P)$-martingale, that is, $\esp{L_T^\Q}=1$. Since the $\bF$-progressively measurable process $\gamma$ can be suitably chosen, to prove this relation we can assume $\gamma \equiv 0$, without loss of generality.  
Set 
\begin{equation}\label{def:alpha}
\alpha_t := \frac{\mu_S P_{t-\tau}-r(t)}{\sigma_S \sqrt{P_{t-\tau}}}, \quad t \in [0,T].
\end{equation}
We observe that since $\phi(t) > 0$, for each $t \in [-L,0]$, in \eqref{eq:BSdyn},
by Theorem \ref{th:sol}, point (i), we have that $P_{t-\tau} > 0$, $\P$-a.s. for all $t \in [0,T]$, so that the process $\alpha=\{\alpha_t,\ t \in [0,T]\}$ given in \eqref{def:alpha} is well-defined, as well as the random variable $L_T^\Q$.
Clearly, $\alpha$ is an $\bF$-progressively measurable process. Moreover, $\int_0^T|\alpha_u|^2 \ud < \infty$ $\P$-a.s., since sample-path continuity of the process $P$ yields the fulfillment of the almost sure boundedness property by $P$; on the other hand, the condition $\phi(t)>0$, for every $t \in [-L,0]$, implies that almost every path of $\left\{\frac{1}{\sigma_S \sqrt{P_{t-\tau}}},\ t \in [0,T]\right\}$ is bounded on the compact interval $[0,T]$.
Set $\F_t^P:=\F_0^P=\{\Omega,\emptyset\}$, for $t \leq 0$. Then, $\alpha_u$, for every $u \in [0,T]$, is $\F_{T-\tau}^P$-measurable. 
Since $Z_{u-\tau}$ is independent of $W_u$, for every $u \in [\tau, T]$, the stochastic integral $\int_0^T \alpha_u \ud W_u$ conditioned on $\F_{T-\tau}^P$ has a normal distribution with mean zero and variance $\int_0^T |\alpha_u|^2 \ud u$. Consequently, the formula for the moment generating function of a normal distribution implies
$$
\esp{\exp{\left(\int_0^T \alpha_u \ud W_u\right)\bigg{|}\F_{T-\tau}^P}}=\exp{\left(\frac{1}{2}\int_0^T |\alpha_u|^2 \ud u\right)},
$$
or equivalently
\begin{equation}\label{eq:gen}
\esp{\exp{\left(\int_0^T \alpha_u \ud W_u-\frac{1}{2}\int_0^T |\alpha_u|^2 \ud u\right)}\bigg{|}\F_{T-\tau}^P}=1.
\end{equation}
Taking the expectation of both sides of \eqref{eq:gen} immediately yields $\esp{L_T^\Q}=1$.
Now, set $\widetilde S_t:= \ds \frac{S_t}{B_t}$, for each $t \in [0,T]$.
It remains to verify that the discounted BitCoin price process $\widetilde S=\{\widetilde S_t,\ t \in [0,T]\}$ is an $(\bF,\Q)$-martingale.
By Girsanov's theorem, under the change of measure from $\P$ to $\Q$, we have two independent $(\bF,\Q)$-Brownian motions $W^\Q=\{W_t^\Q,\ t \in [0,T]\}$ and $Z^\Q=\{Z_t^\Q,\ t \in [0,T]\}$ defined respectively by
\begin{align*}
W_t^\Q & := W_t + \int_0^t\alpha_u \ud u,\quad t \in [0,T],\\
Z_t^\Q & := Z_t + \int_0^t\gamma_s \ud s,
\quad t \in [0,T].
\end{align*}
Under the martingale measure $\Q$, the discounted BitCoin price process $\widetilde S$ satisfies the following dynamics
\begin{align*}
\ud \widetilde S_t & =  \widetilde S_t\sigma_S\sqrt{P_{t-\tau}}\ud W_t^\Q, \quad \widetilde S_0=s_0 \in \R_+,
\end{align*}
which implies that $\widetilde S$ is an $(\bF,\Q)$-local martingale. Finally, proceeding as above it is easy to check that $\widetilde S$ is a true $(\bF,\Q)$-martingale.

\end{proof}

\end{document}